\documentclass[pra,10pt,twocolumn,superscriptaddress]{revtex4}
\usepackage[unicode=true,pdfusetitle, bookmarks=true,bookmarksnumbered=false,bookmarksopen=false, breaklinks=false,pdfborder={0 0 0},backref=false,colorlinks=false] {hyperref}
\hypersetup{ colorlinks,linkcolor=myurlcolor,citecolor=myurlcolor,urlcolor=myurlcolor}
\usepackage{braket,colortbl,amsthm,amsmath,cleveref,amssymb,txfonts,soul}
\definecolor{myurlcolor}{rgb}{0,0,0.7}
\usepackage{graphics,graphicx}
\usepackage{color}
\usepackage{graphicx}
\usepackage[format = plain,labelfont = bf,up, textfont = normal , up, justification =raggedright, singlelinecheck =false]{caption}
\usepackage{tabularx}
\usepackage{amsmath}

\usepackage{graphicx}
\usepackage[utf8x]{inputenc}
\usepackage{color}
\usepackage{braket}
\usepackage{latexsym}
\usepackage{amssymb}
\usepackage{amsthm}
\usepackage{bm}
\usepackage{graphics,epstopdf}
\usepackage{color}
\usepackage{enumitem}
\usepackage{fmtcount}
\usepackage{booktabs}
\usepackage{soul}

\usepackage{epsfig}

\theoremstyle{plain}
\newtheorem{theorem}{Theorem}

\def\bea{\begin{eqnarray}}
\def\eea{\end{eqnarray}}
\def\ba{\begin{array}}
\def\ea{\end{array}}

\def\ket{\rangle}
\def\bra{\langle}
\def\beq{\begin{equation}}
\def\eeq{\end{equation}}

\begin{document}

%\title{Resource theory of coherence and wave-particle duality in a linearly dependent basis}
%\title{Resource theory of coherence in a linearly dependent basis and the resulting wave-particle duality}

%\title{Resource theory of quantum coherence with  linearly dependent pointers\\and corresponding wave-particle duality}

%strongly non-distinguishable pointers

\title{Resource theory of quantum coherence with  probabilistically non-distinguishable pointers\\and corresponding wave-particle duality}

\author{Chirag Srivastava}
\affiliation{Harish-Chandra Research Institute, HBNI, Chhatnag Road, Jhunsi, Allahabad 211 019, India}

\author{Sreetama Das} 
\affiliation{Harish-Chandra Research Institute, HBNI, Chhatnag Road, Jhunsi, Allahabad 211 019, India}
\affiliation{Faculty of Physics, Arnold Sommerfeld Centre for Theoretical Physics (ASC), Ludwig Maximilian University, Theresienstr. 37, 80333 Munich, Germany}

\author{Ujjwal Sen}
\affiliation{Harish-Chandra Research Institute, HBNI, Chhatnag Road, Jhunsi, Allahabad 211 019, India}
\begin{abstract}

%\textcolor{red}{later!!}
 One of the fundamental features of quantum mechanics is the superposition principle, a manifestation of which is embodied in quantum coherence. Coherence of a quantum state is invariably defined with respect to a preferred set of pointer states, and there exist  quantum coherence measures with respect to 
 %. A number of previous works discuss coherence in 
deterministically as well as probabilistically distinguishable sets of quantum state vectors.
% orthogonal as well as non-orthogonal but linearly independent bases. 
Here we study the resource theory of quantum coherence with respect to an arbitrary set of quantum state vectors, that may not even be probabilistically distinguishable. Geometrically, a probabilistically indistinguishable set of quantum state vectors forms a linearly dependent set.
%characterization of such a set is its linear dependence. 
%states that span the  general spanning set, basis consisting linearly dependent states. 
We find the free states of the resource 
theory, and analyze the 
corresponding free operations, 
%and also 
obtaining a necessary condition for an arbitrary quantum operation to be free. %on free/incoherent operations when the basis consists of only three pure qubits. 
We identify 
%discuss about the potential 
a class of measures of the quantum coherence, and in particular establish a monotonicity property of the measures. %measures with respect to any general basis. 
We find a connection of an arbitrary set of quantum state vectors 
%forming a basis, 
with positive operator valued measurements with respect to the resource theory being considered, which %. This 
paves the way for an alternate definition of the free states.
%in a general basis, based on measurements. 
%We also argue that the resource theory, considered here, is different with the existing resource theories of coherence based on measurements. 
We notice that the resource theory of magic can be looked upon as a resource theory of quantum coherence with respect to a set of quantum state vectors that are probabilistically indistinguishable.
We subsequently examine the wave-particle duality in a double-slit set-up in which superposition of 
%linearly dependent states 
probabilistically indistinguishable quantum state vectors
is possible. Specifically, we report a complementary relation between quantum coherence and path distinguishability in such a set-up.

\end{abstract}

\maketitle
\section{Introduction}

%The superseding success of quantum mechanical protocols over the corresponding classical protocols, lies in the extraordinary fact that, the laws of quantum mechanics allow one to achieve certain tasks which are either classically impossible, or are poorly executed, in terms of efficiency. Examples of such novel quantum mechanical protocols are quantum teleportation \cite{ben'93}, quantum dense coding \cite{ben'92}, quantum key-distribution \cite{key} etc. 

A quantum state possesses certain properties which are absent in the states of a classical system.
These properties are the key elements acting as ``resource'' in various quantum information, computational, and communication protocols,
like quantum teleportation \cite{ben'93}, quantum dense coding \cite{ben'92}, and quantum key distribution \cite{key}. It is therefore important to experimentally preserve and theoretically study such resources for their efficient use in 
%Hence, quantifying and preserving resources is crucial for 
quantum technologies. In the last few decades, the study of quantum entanglement \cite{hor'09,ent-2,das'06} as a resource has seen a significant amount of development.
%and established on a concrete ground. 
%Another resource theory, t
Just like entanglement, the concept of quantum coherence was known from the early days of quantum theory. However, a systematic study of the resource theory of quantum coherence  has gained attention only in the last few years, and the field has since experienced considerable progress \cite{aberg'06,plenio'14,winter'16}. Coherence of a quantum state is the manifestation of the superposition principle of quantum mechanics, which says that a quantum system can simultaneously exist in more than one state. It allows one to understand phenomena like the interference pattern in an interferometer.
%,  which were not explicable earlier. 
In quantum computational tasks, it is essential to preserve coherence, as decoherence can lead to loss of quantum properties of the system, consequent to which the advantage of the machine over its classical counterpart may get significantly affected.

Quantum coherence is invariably defined with respect to a fixed preferred basis of the system. %In the pioneering paper by Baumgratz et al. \cite{plenio'14}, 
A resource theoretic description of quantum coherence, viz. the coherence-free states, free operations, and valid measures of coherence,
is typically formulated by using 
%was put forward, using 
an orthogonal basis as the preferred one. Following the initial works on quantification of quantum coherence as a resource \cite{aberg'06,plenio'14,winter'16}, a large number of works contributed to this field, which addressed topics like a systematic examination of the classes of free operations, the convertibility between an arbitrary state and the state containing maximal coherence, and the valid measures of coherence. See \cite{coherence_review} and references therein. It has been possible to relate quantum coherence with other quantum resources, e.g. non-Markovianity \cite{chanda'16,huang'17,cakmak'17,chiru'17,morte'17,pati'18}, non-locality \cite{mondal'17,bu'16,qiu'16}, entanglement \cite{jk'05,stre'15,xi'15,kill'16,stre'16,qi'17,chin'17,zhu'17,zhu'18},
and quantum discord \cite{yao'15,ma'16,guo'17}. Coherence have been proven to be beneficial for cooling in quantum refrigerators \cite{huber'15} and extracting work using quantum thermal machines \cite{anders'16,jenn'16}.

In an double-slit set-up, the states representing the beams coming out of different slits are mutually distinguishable, and are hence represented in the quantum description of the set-up by orthogonal pure states. Such distinguishable ``pointers'' appear 
%The situation is similar 
in other experimental adaptations of the same theoretical description, as in interferometric set-ups, leading to the usual choice of orthogonal bases for describing quantum coherence. Hence the coherence of the state after superposition, can be described in orthogonal basis. A Hilbert space that describes a quantum system can however be equivalently described by linearly independent non-orthogonal states. We remember that a set of mutually orthogonal states is of course linearly independent.  The resource theory of quantum coherence in superpositions over such linearly independent sets of states   
%however can be In a more general scenario, resource theory of coherence with respect to non-orthogonal but linearly independent basis 
has been explored recently in \cite{plenio'17,das'20}, with a double-slit set-up where such a superposition can naturally appear being provided in \cite{das'20}. 
There also exist a different formalism of coherence, based on measurements \cite{aberg'06,bruss'19,cimini'19,bruss2'19}. The motivation of these comes from the fact that an orthogonal set of states constituting a basis, used to define the coherence, can be seen as the outcomes of a rank-1 projective measurement. In  \cite{aberg'06}, a framework is given which can be seen as defining coherence based on higher-rank projective measurements, such that the corresponding relevant block-diagonal states form the set of free states. 
%Based on this idea, a 
A generalized theory of coherence with respect to generalized measurements, i.e., positive operator valued measurements (POVMs), has been proposed in \cite{bruss'19}. Further avenues of defining quantum coherence include Refs. \cite{aro-anekey}.

It may be noted here that while a set of quantum states can be distinguished deterministically if and only if they are mutually orthogonal, such a set can be probabilistically distinguished unambiguously if and only if they are linearly independent 
\cite{ghughu-eseychhilo-janalai}. A Hilbert space describing a quantum system can also be spanned by a set of linearly dependent states. We call such a set of states as a linearly dependent basis. It is clear that the set is overcomplete as a span of the corresponding Hilbert space. 
Physically interesting sets of states that form a linearly dependent basis include the overcomplete set of coherent states of a mode of an electromagnetic field \cite{malhar-jekhane-ananta}. Another scenario where such a situation appears naturally is 
in the resource theory of stabilizer-based quantum computation \cite{victor'14} (see also \cite{eternal-triangle,hopeless-hexagon}). 
It is therefore 
natural to ask: Can we formulate a resource theory of superposition over a linearly dependent basis? And can it be seen as 
a superposition over probabilistically indistinguishable pointers in a double-slit set-up?

%Given the existence of generalization of coherence from an orthogonal basis to a non-orthogonal but linearly independent basis, it is natural to ask- what if the basis consists of linearly dependent states? This question obviously raises another important question- does there exist any practical importance of coherence in linearly dependent basis? Indeed, its practical utility can be seen in the context of the resource theory of stabilizer computation \cite{victor'14}, which can essentially be seen as a resource theory of coherence with respect to some fixed linearly dependent basis. Here, the quantity  `magic' is a resource and is quite helpful for universal quantum computation.

We aim at obtaining two results in this paper. 
Firstly, 
%In this paper, 
similar to every resource theory, we first define the free states and the free operations for the resource theory of quantum coherence with respect to a general, possibly linearly dependent, basis. We also show that the resource theory of magic \cite{victor'14} can be seen as a resource theory of quantum coherence with respect to a linearly dependent basis.
%The next aim is to obtain the form or structure of incoherent operations. In this regard, it is important to mention that it can be a very difficult task to characterize the free operations if arbitrary number of linearly dependent states are considered as basis elements. Therefore, in this article, 
We subsequently provide a necessary condition to characterize free operations with respect to a qubit basis containing 
%arbitrary 
three 
arbitrary 
pure states. We then go on to define quantum coherence measures in the scenario of a general basis. 
%a linearly dependent basis scenario. 
Similar to coherence measures for orthogonal and linearly independent bases, any contractive distance measure between states suffices to provide a suitable measure of coherence in this case also.
% and like in previous resource theories of coherence with respect to orthogonal or linearly independent basis, any contractive distance measure between the states suffice criteria for a suitable candidate for coherence measure with respect to linearly dependent basis, too. 
%\textcolor{red}{(later!!)
We also develop a connection of an arbitrary set of, possibly linearly dependent, states with generalized measurements (POVMs), which helps in developing a relation between  incoherent states in the resource theory of quantum coherence with respect to a general basis and  generalized measurements. We also argue that the resource theory considered here is different from the POVM-based resource theory of quantum coherence developed in Ref. \cite{bruss'19}.
%}

Secondly, we show that it is possible to have multi-slit set-ups where a superposition over linearly dependent states 
can appear naturally.
%Next, 
We study a wave-particle duality relation for a quantum system 
%quanton (quantum particle) 
passing through such a 
double-slit apparatus.
%interferometer.
%ric set-up. 
The corresponding duality relations for superpositions 
over orthogonal and linearly independent bases have been studied earlier in Refs.
%, where a superposition of linearly dependent states is possible. Note that, such duality has been rigorously studied earlier where the basis of superposition consisted of orthogonal or linearly independent pure states 
\cite{das'20,zurek'79,green'88,englert'96,englert'08,bera'15}. In a double-slit set-up, the wave nature of the quantum system is quantified by the quantum coherence of the corresponding state, measured in the basis that is preferred by the apparatus considered. 
%formed by states representing the paths coming out of different slits. 
Whereas, the particle nature can be tracked down by employing a detector in each of these paths, and is quantified using the probability of successfully discriminating the possibly non-orthogonal - but linearly independent - detector states by unambiguous quantum state discrimination (UQSD) \cite{uqsd}.
We introduce 
%an interferometric 
double-slit
set-up where it is possible to superpose three fixed linearly dependent qubit states. We show, through numerical studies, that similar to the cases of superpositions over orthogonal and non-orthogonal but linearly independent bases, we again have a %also exists a 
complementarity relation between quantum coherence and path distinguishability of the quantum system.

%\textcolor{red}{later!! ---- 
The rest of the paper is organized as follows. In Sec. \ref{sec2}, the free states and free operations with respect to a general, possibly linearly dependent, basis are defined. The connection to magic is formally mentioned here. In Sec. \ref{sec3}, a necessary condition on the structure of free operations for the basis formed by any three pure states of a qubit is provided. In Sec. \ref{sec4}, we propose measures of quantum coherence in the case of a general basis.
%linearly dependent basis. 
In Sec. \ref{sec5}, a connection is established between the states of a general basis 
%linearly dependent states 
with the outcomes of a POVM, from the perspective of the resource theory under construction here.
%, and also possibility of superposing any set of states in an interferometric set-up is shown. 
In Sec. \ref{sec6}, we study the wave-particle duality in a set-up where superposition of three linearly dependent qubit states is possible. Finally, we conclude in Sec.\ref{sec7}.
%}

%\section{Definitions}
 \section{Free states and free operations: Link to magic}
 \label{sec2}
One of the fundamental features of quantum mechanics is the superposition principle, i.e., a given pure quantum state can be expressed as a linear superposition of other pure quantum states. For example, the state of the photon after passing through the slits in an Young's double-slit experimental set-up is a superposition of states of the photon passing through either of the slits. These states, representing different paths, form a natural basis for this set-up. Generally considered to be mutually orthogonal, these states can also be mutually non-orthogonal and linearly independent in realistic leaky set-ups \cite{das'20}. Quantum 
%As mentioned before,  
 coherence with respect to mutually orthogonal bases as also the more general case of 
 %and 
 linearly independent bases, have already been studied in detail \cite{aberg'06,plenio'14,winter'16,plenio'17,das'20}.
In the most general scenario, a ``basis'' may have elements which form a linearly dependent set and span the space. 
  Such a basis of a Hilbert space of dimension $d$ can be represented as \(\mathcal{B}=\{|\psi_1\ket,|\psi_2\ket,\ldots,|\psi_n\ket\}\), where  $|\psi_i\ket \in\mathbb{C}^d$ for $i=1,2,\ldots,n$ and $n\geqslant d$. Clearly, elements of the set $\mathcal{B}$ form a linearly independent set only if $n$ equals $d$.
  Otherwise, they form a ``linearly dependent basis''. We always require \(\mathcal{B}\) to form a spanning set of \(\mathbb{C}^d\). 
  
  A note about terminology is in order here. While a ``basis'' of a (linear) space is typically defined in a mathematics course (see e.g. \cite{AmtritavaGupta-Simmons}) as a linearly independent set that spans the space, there is some tradition of tweaking this definition. A classic example is the ``overcomplete basis'' of coherent states of a mode of an electromagnetic field
  \cite{malhar-jekhane-ananta}. A more recent example is the concept of ``unextendible product bases'' 
  in quantum information theory \cite{Catch-em-Young-aar-tar-chhabi}.

\noindent\textbf{{Free states:}} 
A quantum state $\rho$ is incoherent or free with respect to basis $\mathcal{B}$ if it can be expressed as a classical (i.e., probabilistic) mixture of the projectors of the basis elements, i.e.,   
\begin{equation}
\rho=\sum_{i=1}^n p_i |\psi_i\ket\bra\psi_i|,
\label{Eq.1}
\end{equation}
where $\{p_i\}$ forms a probability distribution. Let the set of free states with respect to basis $\mathcal{B}$ be denoted as $\mathcal{F}_{\mathcal{B}}$. 
 The states that cannot be expressed as in 
 %not satisfying 
 Eq. (\ref{Eq.1}) are said to be resourceful or coherent with respect to the basis $\mathcal{B}$.

\noindent\textbf{{Magic is quantum coherence with respect to a linearly dependent basis:}} 
Notice that in this formulation, the resource theory of stabilizer quantum computation \cite{victor'14} can be seen as 
%falls under the  category of 
a resource theory of quantum coherence for a linearly dependent basis.  For the case of single qubit stabilizer quantum computation, the free states or the stabilizer states are the states inside the octahedron, with the eigenstates of the three Pauli operators as its vertices on the Bloch sphere. And any state outside this octahedron is considered as resourceful and termed as ``magic'' states. This can be seen as a resource theory of quantum coherence with the linearly dependent basis being the three pairs of eigenstates of the Pauli matrices.
 
 %The resource theory of stabilizer quantum computation is a  
 %Note that the term resourceful associated with some state is not justified until and unless they provide some advantage over those which are not. Indeed, the advantage of resourceful states in case of resource theory of coherence defined with respect to general basis is unknown, but the motivation behind this is that the resource theories of coherence and magic studied previously are special cases of resource theory that is considered here. And the tasks of resourceful states in those cases are well known. 
 
\noindent\textbf{{Free operations:}}  Any valid physical operation, mapping each state from the set of free states to another free state is defined as a free operation, or incoherent operation. In quantum mechanics, any physical operation can be implemented by a completely positive and trace preserving (CPTP) map. Therefore, a CPTP map, \(\Phi:\mathbb{B}_+(\mathbb{C}^d) \rightarrow \mathbb{B}_+(\mathbb{C}^d)\), is said to be free or incoherent if $\Phi(\rho_f)\in\mathcal{F}_{\mathcal{B}}, \forall \rho_f \in \mathcal{F}_{\mathcal{B}}$, where 
\(\mathbb{B}_+(\mathbb{C}^d)\) is the space of all hermitian operators on \(\mathbb{C}^d\) that have a nonnegative spectrum. 
Any CPTP map $\Phi$, can be represented by a set of Kraus operators, ${K_i}$, that is 
\begin{equation}
\label{Eq2}
\Phi(\rho)=\sum_i  K_i \rho K_i^\dagger,
\end{equation}
with $\sum_i K_i^\dagger K_i = \mathbb{I}_d$, where \(\mathbb{I}_d\) denotes the identity operator on \(\mathbb{C}^d\), and \(\rho \in \mathbb{B}_+(\mathbb{C}^d)\) and has unit trace. Note that the map $\Phi$ is not guaranteed to return a free state if the post-selection of a particular Kraus operator or a particular set therefrom is permitted. The set of all CPTP operations satisfying Eq. (\ref{Eq2}) form a maximal set of incoherent operations and are denoted as $\mathcal{MIO}$.

 Consider now the incoherent operations which remain incoherent under any post-selection. In this case, 
 %it is enough to 
 we 
 demand that
 \begin{equation}
 \label{Eq3}
\frac{K_i \rho_f K_i^\dagger}{\text{tr}\{K_i \rho_f K_i^\dagger\}} \in \mathcal{F_B},
 \end{equation}
for all $i$ and for all \(\rho_f \in \mathcal{F_B}\).  The set of operations satisfying Eq. (\ref{Eq3}) are denoted as $\mathcal{IO}$.  
 Clearly, $\mathcal{IO}$ forms a subset of $\mathcal{MIO}.$ 

\section{Structure of Incoherent operations}
\label{sec3}
An important task in any resource theory is to characterize the structure of free operations. 
The most well-known example is probably the set of 
%We know that, the 
local quantum operations and classical communication (LOCC), which are  free operations in the resource theory of entanglement:
While the operational meaning is known, 
a mathematically precise form is lacking for the set. 
%, but still LOCCs lack a mathematical characterization. 
For the resource theory of quantum coherence with respect to mutually orthogonal or linearly independent bases, a successful characterization of some of the incoherent operations has been possible \cite{aberg'06,plenio'14,winter'16,plenio'17,das'20}. 
While the complete characterization of the set of free operations in a resource is often difficult, it is often useful to consider supersets of the required set. An example is the set of separable superoperators \cite{cha-garam}, which forms a superset of LOCC \cite{Badami}. In Theorem 1 below, we similarly identify a superset of the set of free operations in the resource theory of quantum coherence with an arbitrary linearly dependent qubit basis of cardinality \(= 3\).

%For resource theory of coherence with a basis formed by an arbitrary number of linearly dependent states, obtaining the form of incoherent operations is comparatively difficult. A bit simplified version of the problem is when one considers a fixed number of linearly dependent states. Here we take a simple example of resource theory of coherence with a basis consisting of three qubits. 
Towards this end, we consider \(\mathcal{B}_3=\{|\psi_1\ket,~|\psi_2\ket,|\psi_3\ket\}\) as an arbitrary linearly dependent qubit basis. Since any three (distinct) pure qubit states always lie on some circle (though not necessarily a great circle) on the surface of the Bloch sphere, without loss of generality, we can write \(|\psi_1\ket=\cos\frac{\theta}{2}|0\ket + e^{i\phi_1}\sin\frac{\theta}{2}|1\ket\), \(|\psi_2\ket=\cos\frac{\theta}{2}|0\ket + e^{i\phi_2}\sin\frac{\theta}{2}|1\ket\), and \(|\psi_3\ket=\cos\frac{\theta}{2}|0\ket + e^{i\phi_3}\sin\frac{\theta}{2}|1\ket\) , where $|0\ket$ and $|1\ket$ are orthogonal to each other. Note that no two of the $\phi_1, \phi_2,  \phi_3$ are equal, and without loss of generality, $\phi_1$ can be put equal to zero. Also, $\theta=0$ and $\theta=\pi$ are not allowed. 
Let $K$ be one of the Kraus operators constituting the corresponding set of incoherent operations $\mathcal{IO}$ for the basis $\mathcal{B}_3$. Since these incoherent operations map states from $\mathbb{B}_+(\mathbb{C}^2)$ to $\mathbb{B}_+(\mathbb{C}^2)$, $K$ can be expanded as $\sum_{m,n=0}^{1}K_{mn}|m\ket\bra n|$. The following theorem gives a necessary condition that each Kraus operator  constituting any incoherent operation belonging to $\mathcal{IO}$ for the basis $\mathcal{B}_3$ must satisfy.
\begin{theorem}
A necessary condition that a Kraus operator $K$ is a member of the set of Kraus operators of an incoherent operation, belonging to $\mathcal{IO}$, with respect to the qubit basis $\{|\psi_1\ket,~|\psi_2\ket,|\psi_3\ket\}$ of cardinality \(= 3\), is given by
\begin{equation}
\label{Eq4}
|K_{11}|^2 - |K_{22}|^2 + |K_{12}|^2 \frac{1}{\kappa}
%\frac{K^*_{21}K_{22}}{K^*_{11}K_{12}} 
- |K_{21}|^2 \kappa, %\frac{K^*_{11}K_{12}}{K^*_{21}K_{22}}=0.
\end{equation}
where \(\kappa = \frac{K^*_{11}K_{12}}{K^*_{21}K_{22}}\).
\end{theorem}
\begin{proof} Let $\rho$ be a free state in the resource theory of quantum coherence with respect to the linearly dependent basis, $\{|\psi_1\ket,|\psi_2\ket,|\psi_3\ket\}$. Let $K$ be the Kraus operator such that
\begin{equation}
\label{Eq5}
\frac{K \rho K^\dagger}{\text{tr}\{K \rho K^\dagger\}}=\rho',
\end{equation}
 where $\rho'$ is also an incoherent state. The free states can be written as
 $\rho=p_1|\psi_1\ket\bra\psi_1| + p_2|\psi_2\ket\bra\psi_2| +(1-p_1-p_2)|\psi_3\ket\bra\psi_3|$ and $\rho'=p'_1|\psi_1\ket\bra\psi_1| + p'_2|\psi_2\ket\bra\psi_2| +(1-p'_1-p'_2)|\psi_3\ket\bra\psi_3|,$ for $0\leqslant p_1$, $p_2$, $p'_1$, $p'_2$, $p_1+p_2$, $p'_1+p'_2\leqslant 1$.  Comparing any of the diagonal entries on both sides of Eq. (\ref{Eq5}) gives 
 \begin{eqnarray}
 \label{Eq6}
 |K_{11}|^2-|K_{22}|^2+|K_{12}|^2\tan^2\frac{\theta}{2} - |K_{21}|^2\cot^2\frac{\theta}{2} + 2\cot\frac{\theta}{2}\text{Re}(\Delta e^{i\phi_3})\nonumber\\
 +2\cot\frac{\theta}{2}\text{Re}(\Delta (1-e^{i\phi_3}))p_1
 +2\cot\frac{\theta}{2}\text{Re}(\Delta (e^{i \phi_2} - e^{i\phi_3}))p_2=0,\nonumber  \\
 \end{eqnarray}
where $\Delta=K^*_{11}K_{12}\tan^2\frac{\theta}{2}-K^*_{21}K_{22}.$ 
Since Eq. (\ref{Eq6}) is of the form $A+Bp_1+Cp_2=0$, and this must hold for every $p_1$ and $p_2$, therefore the only solution is given by $A=B=C=0.$
One can show that $B=0$ and $C=0$ imply $\Delta=0$ for nonzero $\phi_2, \phi_3$ and $\theta\neq 0$. And $\Delta =0$ implies  $|K_{11}|^2 - |K_{22}|^2 + |K_{12}|^2\frac{K^*_{21}K_{22}}{K^*_{11}K_{12}} - |K_{21}|^2\frac{K^*_{11}K_{12}}{K^*_{21}K_{22}}=0.$
\end{proof}

\section{Measures of quantum coherence}
\label{sec4}
Along with identifying the set of free states and free operations, an important aspect of a resource theory is also to recognize potential measures of the resource.
%defined in any resource theory, the next important step is to look out for resource quantifiers. 
One of the crucial properties of a resource quantifier is that it never increases under the corresponding set of free operations. For example, an 
%entanglement of formation or any other 
entanglement measure is typically an LOCC monotone, i.e., non-increasing 
%of a state never increases 
under LOCC. 
%There also exist quite a number of coherence measures in resource theory of coherence with orthogonal or linearly independent basis \cite{plenio'14,winter'16,plenio'17,das'20}.  
%In this paper, w
We propose here a quantum coherence measure, $C_\mathcal{B}$, defined with respect to a general basis $\mathcal{B}$, which could be a linearly dependent one. 
Important properties which our coherence measure should satisfy are 
\begin{itemize}
\item[1.] $C_{\mathcal{B}}(\rho) \geq 0$, and the equality holds  iff $\rho \in \mathcal{F_{\mathcal{B}}}$,
\item[2.] $C_{\mathcal{B}}(\Phi(\rho))\leqslant C_{\mathcal{B}}(\rho)$, where $\Phi$ is an incoherent map.
\end{itemize}

We define the quantum coherence measure with respect to the basis $\mathcal{B}$ for a state $\rho$ as 
\begin{equation}
\label{Eq7}
C_{\mathcal{B}}(\rho)=\min_{\sigma \in \mathcal{F}_\mathcal{B}} \mathcal{D}(\rho,\sigma),
\end{equation} 
where $\mathcal{D}$ is a contractive distance measure between states. 
A distance measure \(\mathcal{D}\) on \(\mathbb{B}_+(\mathbb{C}^d)\) is said to be contractive if \(\mathcal{D}(\Phi(\varrho), \Phi(\varsigma)) \leqslant \mathcal{D}(\varrho, \varsigma)\) for all unit trace \(\varrho, \varsigma \in \mathbb{B}_+(\mathbb{C}^d)\) and for all CPTP \(\Phi\).
%Distance measure between two states which is contractive under incoherent CPTP operations are suitable coherence measures, as the satisfy both the properties, 1 and 2.   
The relative entropy 
%of coherence 
is an example of a contractive measure, although it does not satisfy symmetry and the triangle inequality \cite{atma-nirbhar-mane-single}. 
Note that \(\mathcal{F}_{\mathcal{B}}\) is a closed set, and this gives us the right of using a minimum instead of an infimum in the definition of the measure. 

\begin{theorem}
The quantum coherence measure \(C_\mathcal{B}\) with respect to a general basis \(\mathcal{B}\) satisfies the properties 1 and 2.
\end{theorem}

\begin{proof}
The nonnegativity of the measure follows from the nonnegativity of the distance measure. If \(\rho \in \mathcal{F}_\mathcal{B}\), then the measure vanishes by definition. If \(C_\mathcal{B}(\rho)=0\), then \(\rho\) must be a free state, since the set of free states is a closed set in the space \(\mathbb{B}_+(\mathbb{C}^d)\). 

For proving property 2, suppose that the incoherent state that attains the minimum distance in the quantum coherence measure for the state \(\rho\) is \(\tilde{\sigma}\), so that \(C_\mathcal{B}(\rho) = \mathcal{D}(\rho,\tilde{\sigma})\). 
%But, due to the contractive property of \(\mathcal{D}\), we have, 
Therefore,
for all incoherent \(\Phi\), we have
\begin{eqnarray}
C_\mathcal{B}(\rho) & = & \mathcal{D}(\rho,\tilde{\sigma}) \nonumber\\
& \geqslant & \mathcal{D}(\Phi(\rho),\Phi(\tilde{\sigma})) \nonumber\\
& \geqslant & C_\mathcal{B}(\Phi(\rho)), \nonumber
\end{eqnarray}
where the first inequality is due to the contractive nature of \(\mathcal{D}\), while the second inequality is due to the fact that \(\Phi(\tilde{\sigma})\) is an incoherent state.
\end{proof}

For the case of a qubit, the trace distance, which is the trace norm of the difference between two states (density matrices), equals the Euclidean distance between them in the Bloch ball. It is also contractive under trace preserving operations  \cite{nielsen}. Therefore, a suitable and easily computable candidate for a quantum coherence measure could be given by
\begin{equation}
\label{Eq8}
\widetilde{C}_{\mathcal{B}}(\rho)=\min_{\sigma \in \mathcal{F}_{\mathcal{B}}} \text{tr}|\rho-\sigma|,
\end{equation} 
where tr$|\cdot|$ denotes the trace norm of its argument.
Note that the set of maximally coherent states 
%in case of qubits with respect to this measure 
may or may not be a singleton set depending on the states constituting the basis. For example, the only maximally coherent state with respect to  the basis $\{|0\ket, |1\ket, |+\ket_x, |-\ket_x, |+\ket_y\}$ is $|-\ket_y$, where $|0\ket$ and $|1\ket$ are eigenstates of the Pauli operator $\sigma_z$, $|\pm\ket_x$ are eigenstates of the Pauli operator $\sigma_x$, and $|\pm\ket_y$ are eigenstates of the Pauli operator $\sigma_y$. And if the basis is $\{|0\ket, |1\ket, |+\ket_x\}$, then all the (infinite number of) states falling on one-half of the great circle (part of the circle with negative x-values) including states $|+\ket_y$, $|-\ket_x$, $|-\ket_y$ are maximally coherent. This has been illustrated in Fig. \ref{fig2}. 
If the basis is chosen to be $\{|0\ket, |1\ket, |+\ket_x, |-\ket_x, |+\ket_y, |-\ket_y\}$, so that the quantum coherence is equivalent to magic, the number of maximally coherent states is nonunique but finite.
The variation in the number of maximally coherent states with the variation in the basis is easy to understand as a geometric fact using the Bloch ball representation. However, it may be  less expected as an algebraic fact, as when a linearly independent non-orthogonal basis is considered for a qubit, then there exists a unique maximally coherent state, and on the other hand, if the basis is orthogonal, then there exist infinitely many maximally coherent states.

\begin{figure}[h]
%\begin{center}
\includegraphics[width = 0.48\textwidth]{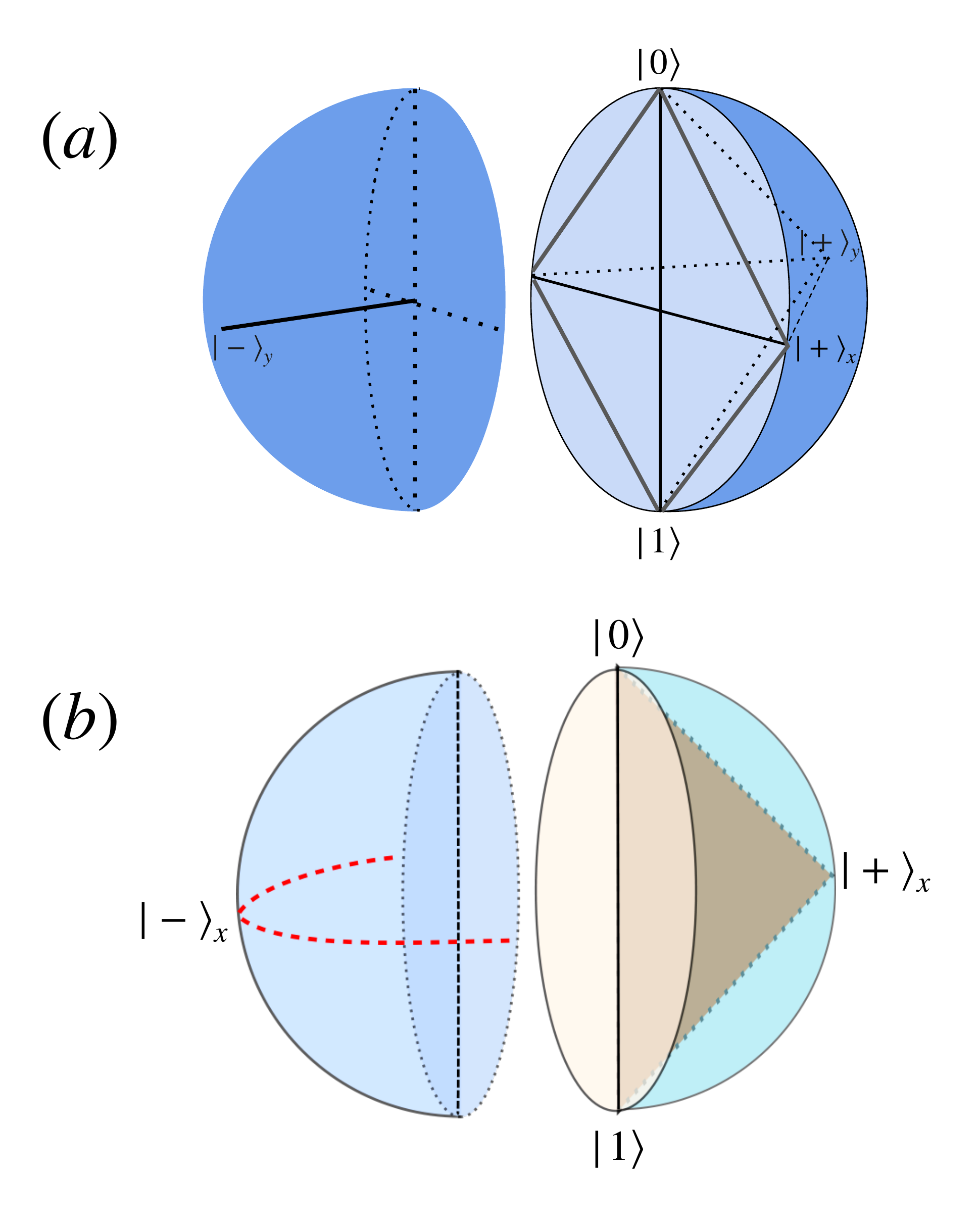} 
\caption{
%Schematic pictures of 
Representations on the Bloch sphere of maximally quantum coherent states 
%Bloch sphere showing the maximally coherent states 
with respect to different linearly dependent bases for a qubit. (a) The linearly dependent basis is $\{|0\rangle,|1\rangle,|+\rangle_{x},|-\rangle_{x},|+\rangle_{y}\}$, as marked on the right hemispherical section of the Bloch sphere, and the unique maximally coherent state is $|-\rangle_{y}$, marked on the left hemispherical section. (b) The linearly dependent basis is $\{|0\rangle, |1\rangle, |+\rangle_{x}\}$ and is shown on the right hemispherical section. There are an infinite number of maximally coherent states in this case. In the left hemispherical section, all the states lying on the red dashed curve are maximally coherent with respect to this basis.}
\label{fig2}
%\end{center}
\end{figure}

\section{Measurements producing Incoherent states }
\label{sec5}
 
%An orthogonal set of states forming a basis can be seen as the outcomes of a rank-1 projective measurement. Therefore, the free states in case of resource theory of coherence with respect to such a basis, are diagonal with respect to the outcomes of the measurement. There exist a generalization of resource theory of coherence based on higher rank projective measurements, in which the free states are block-diagonal, where each outcome of the measurement represents a block. These formulations are further generalized to more general measurements, i.e., POVMs. Here, our task is to compare the resource theories of coherence, defined with respect to a linearly dependent basis and defined based on generalized measurements. 
Any rank-1 (repeatable) projective measurement produces a diagonal state in an orthogonal basis if the result of the measurement is forgotten, i.e., if the classical flags corresponding to the pointer states in the measurement are dumped (traced out) \cite{sedin-chander-alo}. This creates an incoherent state of the resource theory of quantum coherence with respect to the orthogonal basis corresponding to the projection measurement.
We now show how the theory of generalized quantum measurements (POVMs) provides a method for generating 
%In this section, we provide a measurement which can produce an 
incoherent states for a resource theory with respect to a general basis, possibly a linearly dependent one, of an arbitrary dimensional Hilbert space.   
The proof goes as follows. We first show that for every set of states spanning a Hilbert space, there exists a measurement with the same states as its outcomes. This is the content of Theorem 3 below, and assumes the generalized von Neumann-L{\" u}ders postulate \cite{sedin-chander-alo} for the post-measurement states in a repeatable generalized measurement. The incoherent states in the resource theory of quantum coherence, defined with respect to the general basis (the spanning set), can then be obtained as outcomes of the measurement identified in Theorem 3, if we forget the outcome of the measurement. The theorem provides a method of creating particular mixtures of the elements of the chosen basis. Other mixtures of the elements of the same basis can be created by suitably blocking, completely or partially, some or all the outcomes, before the forgetting stage.
%seen as the outcomes of some measurements (POVMs).
%Further, in Young's double slit interferometric set-up the states that superpose are orthogonal.Here, we see that how this idea of a set of states, spanning a Hilbert space,  being a measurement outcome can be used to superpose linearly dependent or linerly independent set of states also, in an interferometric set-ups. 

We note here that there already exists a resource theory of quantum coherence based on an arbitrary POVM \cite{bruss'19}. As we will see later, the resource theory of quantum coherence considered here is different from the POVM-based resource theory.
%, which will be explained later.

%In the following theorem, it is shown that for any set of states, forming a basis, there exist a POVM with the same states as its outcome.   
\begin{theorem}
\label{th2}
For any set $\{|\psi_i\ket\}_{i=1}^n$, where $n\geqslant d$, of the Hilbert space $\mathbb{C}^d$, that spans $\mathbb{C}^d$, it is always possible to find a POVM that has $|\psi_i\ket$ as its outcomes, with possibly some outcomes that need to be ignored.
\end{theorem}
\begin{proof}
If $\sum_{i=1}^n p_i|\psi_i\ket\bra\psi_i|$, $p_i> 0$, $\sum_{i=1}^n p_i=1$, contains the maximally mixed state $\frac{1}{d}\mathbb{I}_d$, then the required POVM elements are $A_i=\sqrt{d\tilde{p}_i}|\psi_i\ket\bra\psi_i|$, where $\tilde{p}_i>0$, $\sum_{i=1}^n \tilde{p}_i=1$, and are such that $\sum_{i=1}^n \tilde{p}_i|\psi_i\ket\bra\psi_i| = \frac{1}{d}\mathbb{I}_d.$ Note that the POVM elements, $A_i$, satisfy  $\sum_{i=1}^n A_i^\dagger A_i=\mathbb{I}_d$.
Also, note that no probability $\tilde{p}_i$ can be zero as it will cause a corresponding POVM element to vanish, therefore leading to the non-occurrence of some $|\psi_i\ket$ as the POVM outcome. Note that we are only considering the cases for which \(p_i \ne 0 \forall i\).

If $\sum_{i=1}^n p_i|\psi_i\ket\bra\psi_i|$ \((p_i \ne 0 \forall i)\) does not contain  $\frac{1}{d}\mathbb{I}_d$, we can always extend the set $\{|\psi_i\ket\}_{i=1}^n$ to $\{|\psi_i\ket\}_{i=1}^m$, $m>n$ such that $\sum_{i=1}^m p_i|\psi_i\ket\bra\psi_i|$ contains $\frac{1}{d}\mathbb{I}_d$. Note that we still have \(p_i \ne 0 \forall i\), and now, \(i = 1, 2, \ldots, m\). We can then choose the POVM elements as $A_i=\sqrt{d\tilde{p}_i}|\psi_i\ket\bra\psi_i|$, where $\tilde{p}_i$ are such that $\sum_{i=1}^m \tilde{p}_i|\psi_i\ket\bra\psi_i| = \frac{1}{d}\mathbb{I}_d$, \(\sum_{i=1}^m \tilde{p}_i =1\) and \(\tilde{p}_i \geq 0\) for all \(i = 1, 2, \ldots, m\).  

In both cases, for the output states to be 
\(|\psi_i\rangle\), for \(i=1,2, \ldots n\), we need to feed the measurement with the completely depolarized state (maximally mixed state), \(\frac{1}{d}\mathbb{I}_d\), or any other state of full rank, and assume the von Neumann-L{\" u}ders measurement postulate, so that the outcomes are 
\begin{equation}
\frac{A_i \frac{1}{d}\mathbb{I}_d A_i^{\dagger}}{\mbox{tr}\left(A_i \frac{1}{d}\mathbb{I}_d A_i^{\dagger}\right)} 
%=
%\frac{A_i}{\mbox{tr}\left(A_i\right)} 
= |\psi_i\rangle \langle \psi_i|.
\end{equation}

In the latter case, the outcomes corresponding to the POVM elements $A_i$ for $i=n+1, n+2,\ldots,m$, need to be ignored.
\end{proof}
To illustrate the result, consider the set $\{|\phi_1\ket, |\phi_2\ket, |\phi_3\ket\}$ spanning $\mathbb{C}^2$, with $|\phi_1\ket=|0\ket$, $|\phi_2\ket=\frac{1}{\sqrt{2}}(|0\ket+|1\ket)$, $|\phi_3\ket=|1\ket$. As $\sum_i p_i|\phi_i\ket\bra\phi_i|$ do not contain $\frac{1}{2}\mathbb{I}_2$ for $p_i>0$, $\sum_i p_i=1$ for $i=1,2$, and 3, we extend the set by adding $|\phi_4\ket=\frac{1}{\sqrt{2}}(|0\ket-|1\ket)$ in the set, and now, $\sum_{i=1}^4 p_i|\phi_i\ket\bra\phi_i|$ contains the maximally mixed state, with \(p_i > 0 \forall i = 1,2,3,4\) and \(\sum_{i=1}^4p_i =1\). 
To create a mixture of \(|0\rangle\), \(\frac{1}{\sqrt{2}}(|0\rangle + |1\rangle)\), and \(|1\rangle\), which is an incoherent state in the resource theory of quantum coherence with respect to the linearly dependent basis formed these three vectors, we can now feed the state \(\frac{1}{2}\mathbb{I}_2\) to the POVM with the POVM elements 
%The POVM element corresponding to the outcome 
%$|\phi_i\ket$ is 
$\sqrt{2\tilde{p}_i}|\phi_i\ket\bra\phi_i|$, \(i=1,2,3,4\).  Note that  $\sum_{i=1}^4 \tilde{p}_i|\phi_i\ket\bra\phi_i|= \mathbb{I}_2$, with the values of the $\tilde{p}_{i}$'s being such that $\tilde{p}_i>0 \forall i = 1,2,3,4$, $\sum_i \tilde{p}_i=1$. The required state is created if we ignore the outcome corresponding to the POVM element with \(i=4\), and carry on to the forgetting stage. The other incoherent states of the same resource theory can be created in the same measurement, but we then have to add a partial or full blocking of some of the non-ignored outcomes, i.e., the outcomes \(i=1, 2,3\), just before the forgetting stage.

%To get the incoherent states 
%for $i=0,1,2,3$.  

\emph{Difference with POVM-based resource theories of quantum coherence:} For the resource theory of quantum coherence with respect to the set $\{|\psi_i\ket\}_{i=1}^n$, which spans $\mathbb{C}^d$, the  incoherent states may be given by the characterization as in Section \ref{sec2}, viz. they are all states that can be expressed as $\sum_i p_i|\psi_i\ket\bra\psi_i|$, where $p_i \geqslant 0$, $\sum_i p_i = 1$. The parallel and equivalent avenue of characterizing the same set of incoherent states 
%other way of defining these incoherent states 
is in terms of measurement operator $\{A_i\}_{i=1}^r$, where $r=n$ or $m$, as described above in this section. 
%according to the above example. 
The incoherent set of states for a set $\{|\psi_i\ket\}_{i=1}^n$ are then given by
\begin{equation}
\sum_{i=1}^n q_i \frac{ A_i \frac{1}{d}\mathbb{I}_d A^{\dagger}_i}{\text{tr}\left(A_i \frac{1}{d}\mathbb{I}_d A_i^{\dagger}\right)}, 
 \label{pawn}
\end{equation}
 %where $\sigma$ is any state acting on $\mathbb{C}^d.$ 
where \(\{q_i\}_{i=1}^n\) forms the probability distribution representing the process of partial blocking just before the forgetting stage, as described above in this section.

%Corresponding to the POVM $\{A_i\}^r_{i=1}$, where $r=n$ or $m$, Naimark theorem guarantees that there is a PV on $\mathbb{C}^d\otimes \mathbb{C}^{d'}$, where $d'$ is the smallest integer such that $dd' \geqslant r$, for which the oucomes in $\mathbb{C}^d$ are on $\{A_i\}_{i=1}^r$. If $dd'>r$, then there will also be $dd'-r$ outlets in the PV measurement that will never click. Let us visualize this PV as an $r$-slit Young-type apparatus for a beam of two particle quantum systems of $\mathbb{C}^d\otimes \mathbb{C}^{d'}$. We now block the slits corresponding to $i=n+1, n+2,\ldots,m$, if $r>n$. We then send a beam of particles of the quantum system $\mathbb{C}^d$ onto the $r$-slit Young-type apparatus. The output on the other side of the slits will be a superposition on the states $\{|\psi_i\ket\}_{i=1}^n$. Therefore, it is possible to superpose a linearly independent or linearly dependent set of states as well.

%Note that t
The above treatment of resource theory of quantum coherence, defined with respect to a set of states forming a basis (spanning set), is quite different from the resource theory of coherence based on generalized measurements (POVMs) considered in Ref. \cite{bruss'19}.
The conceptualizations of the two resource theories are complementary in nature.
The theory based on POVMs, first defines the quantum coherence measure of a state by resorting to a 
%certain 
Naimark extension of the POVM to a projective measurement (PV) in some higher-dimensional Hilbert space. And then, the free states of the theory are defined using the coherence measure.  Whereas, in the resource theory of quantum coherence defined here with respect to a set of states forming a basis, we first define the free states, then define quantum coherence measures based on distance of a state to the set of incoherent states.
%measures.  

An operational difference can be reported by providing  an example in which the free states in the POVM-based theory form a null set, i.e., there are no free states. Such a situation can never appear in the resource theory of quantum coherence considered here. 
%whose outcomes contain a certain set of states, are different from the free states in the resource theory of coherence defined with respect to the same set of states. 
Let us consider the set formed by the  states   $|\chi_0\ket=|0\ket,~|\chi_1\ket=\cos\frac{\theta}{2}|0\ket+\sin\frac{\theta}{2}|1\ket,~|\chi_2\ket=\cos\frac{\theta}{2}|0\ket-\sin\frac{\theta}{2}|1\ket.$ Note that this set of states falls on a great circle on the Bloch sphere, and for $\theta > \frac{\pi}{2}$, $\sum_i p_i|\chi_i\ket\bra\chi_i|$ contains the maximally mixed state, $\frac{1}{2}\mathbb{I}_2$, where $p_i>0$ and $\sum_i p_i = 1$.
Let $\{B_i\}_{i=0}^2$ be the POVM, acting on $\mathbb{C}^2$, where for every $i$, $B_i$ corresponds to an operator proportional to the projector of 
%the outcome 
$|\chi_i\ket$. The condition $\sum_i B_i^\dagger B_i= \mathbb{I}_2$ is met when 
\begin{eqnarray}
B_0&=&\sqrt{1-\cot^2\frac{\theta}{2}}|\chi_0\ket\bra\chi_0|, \nonumber \\
B_1 &=& \sqrt{\frac{1}{2}\mathrm{cosec}^2\frac{\theta}{2}}|\chi_1\ket\bra\chi_1|, \nonumber \\ 
B_2 &=& \sqrt{\frac{1}{2}\mathrm{cosec}^2\frac{\theta}{2}}|\chi_2\ket\bra\chi_2|. \nonumber
\end{eqnarray}
Let $\rho$ be a quantum state acting on $\mathbb{C}^2$. This state is incoherent in the POVM-based resource theory of quantum coherence, for the POVM $\{B_i\}_{i=0}^2$ iff \cite{bruss'19}
\begin{equation}
B_i^\dagger B_i \rho B_j^\dagger B_j \neq 0, ~ \forall i\neq j.
\end{equation}
It can be shown that the above conditions lead to the requirement $\bra\chi_0|\rho|\chi_1\ket= \bra\chi_1|\rho|\chi_2\ket= \bra\chi_2|\rho|\chi_0\ket=0,$ which is not possible for any state on $\mathbb{C}^2$. 

%Thus, there does not exist any free state based on POVM $\{B_i\}_{i=0}^2$, but in case of resource theory with respect to basis $\{|\chi_i\ket\}_{i=0}^2$, the free states are given by $\sum_i p_i |\chi_i\ket\bra \chi_i|$, where $p_i\geqslant0$, $\sum_i p_i=1$.

%\begin{theorem}
%For any rank-1 POVM $\{B_i\}_{i=1}^n$ with the completeness relation, $\sum_{i=1}^n B_i^\dagger B_i=\mathbb{I}_d$, and where $n>d$,  it is always possible to find a linearly dependent set of quantum states which form outcomes of the POVM. 
%\end{theorem}
%\begin{proof}
%Since, $B_i$ are rank 1, therefore these imply $B_i=c_i|\psi_i\ket\bra\psi|$, where $c_i$ are non-zero complex numbers and $|\psi_i\ket$ are some pure quantum states.
%Putting the form of $B_i$ in the completeness relation gives 
%\begin{equation}
%\label{Eq9}
%\sum_{i=1}^n \frac{|c_i|^2}{d}|\psi_i\ket\bra\psi_i|=\frac{1}{d}\mathbb{I}_d. 
%\end{equation}
%This imply, the set of states $|\psi_i\ket$ spans Hilbert space $\mathbb{C}^d$. Therefore for $n=d$, $\{|\psi_i\ket\}_{i=1}^n$ form an orthogonal set of states, and otherwise for $n>d$, a linearly dependent set.  
%\end{proof}

 \section{Complementary between coherence and path distinguishability}
 %in three-way interferometer}
 \label{sec6}

\begin{figure}[h]
\begin{center}
\includegraphics[width = 0.5\textwidth]{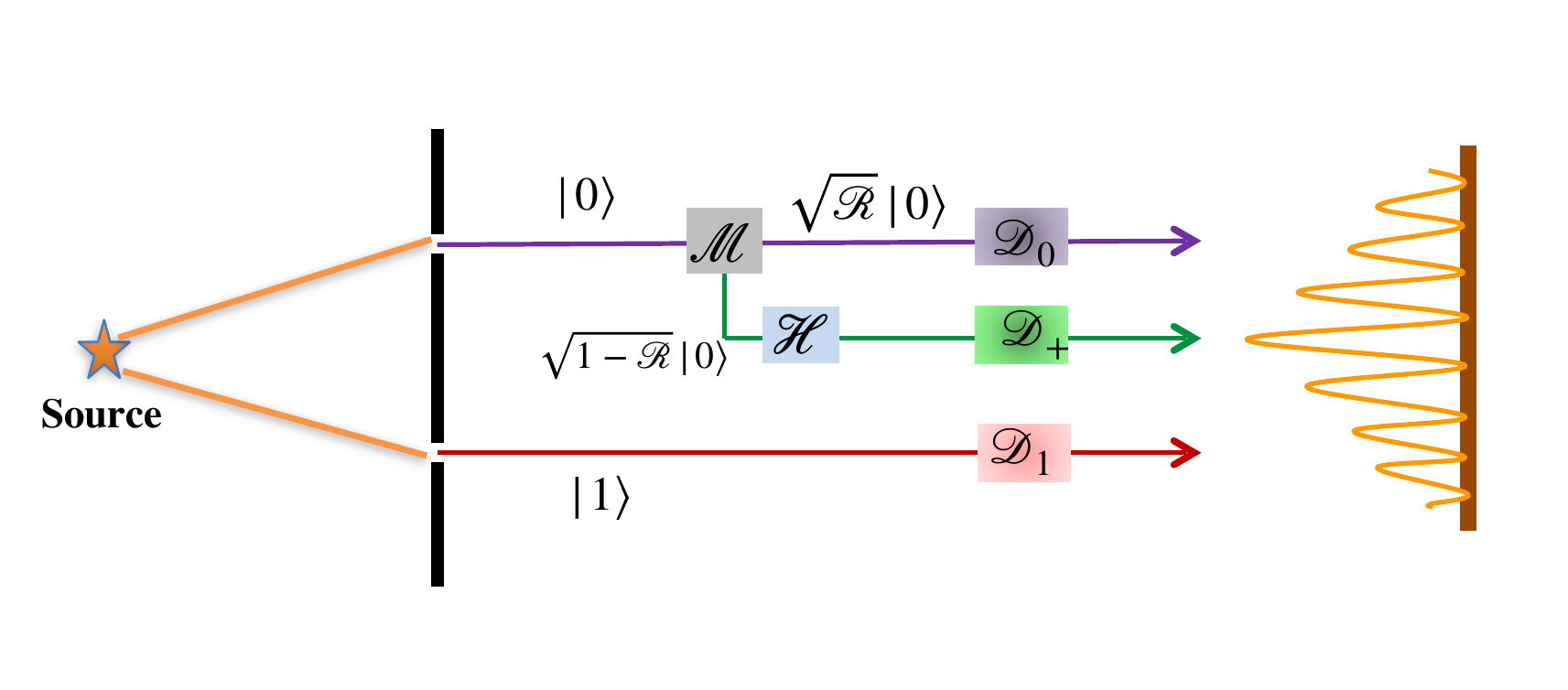} 
\caption{Wave-particle duality in a double-slit apparatus that superposes a linearly dependent set of states. This is a schematic diagram of a  double-slit 
%interferometric 
set-up to superpose the linearly dependent set of states, $\{|0\ket,~|+\ket,~|1\ket\}$. 
The state of the photon, passing through the uppermost slit is denoted by $|0\ket$, whereas, of that passing through lowermost slit is denoted by $|1\ket$, which is orthogonal to $|0\ket$.
The device $\mathcal{M}$ is  inserted in the path of $|0\rangle$ allowing a ``leaked'' part of the photon in the uppermost beam to pass through 
%upper slit 
%to pass through 
a Hadamard gate, marked in the figure as $\mathcal{H}$, and thus obtaining the state $|+\ket$. The non-leaking probability is denoted by $\mathcal{R}$. The three detectors $\mathcal{D}_0$, $\mathcal{D}_+$, and $\mathcal{D}_1$ are inserted in the paths for  $|0\ket,~|+\ket$, and $|1\ket$, respectively, before they impinge on a screen for observation of the interference pattern.
%At last, there is a screen to observe the interference.
}
\label{fig1}
\end{center}
\end{figure}

 In this section, we inquire about how the complementarity relation between coherence and path distinguishability is affected, when an interference between three linearly dependent states takes place in a double-slit apparatus.
 It is possible to 
 %similarly 
 consider a parallel apparatus for the same purpose in an interferometric set-up.
 %an interferometric set-up. 
 We conceptualize a set-up like the one schematically pictured in Fig. \ref{fig1}. 
A double slit is used to provide two different paths forming the quantum system, as considered in Young's double-slit experiment \cite{feynman}. The upper path, in Fig. \ref{fig1}, is represented by the state $|0\ket$, whereas the lower path is represented by the state $|1\ket$, with $|0\ket$ and $|1\ket$ being orthogonal, and spanning the Hilbert space $\mathbb{C}^2$. 
% Unlike in the standard Young's double slit experiment \cite{feynman}, here we don't allow these different path states to interfere at this stage.
 Let the input state be $\alpha|0\ket+\beta|1\ket$, with \(\alpha\), \(\beta\) being arbitrary complex numbers satisfying $|\alpha|^2+|\beta|^2=1$. 
In the next step, we insert a device $\mathcal{M}$ in the path represented by state $|0\rangle$ (i.e., the upper path), which allows the state $|0\rangle$ to pass uninterrupted with probability $\mathcal{R}$, and with probability $1-\mathcal{R}$, it makes the state $|0\rangle$ to get deflected and pass through a Hadamard gate $\mathcal{H}$, and pass on through a ``middle'' path, so that state vector representing the middle path is 
%$|0\ket \rightarrow 
\(|+\ket \equiv \frac{|0\ket + |1\ket}{\sqrt{2}}\). It is to be noted that the Hadamard gate transforms the states  \(|0\rangle\) and \(|1\rangle\) to \(|\pm\rangle \equiv \frac{1}{\sqrt{2}}(|0\rangle \pm |1\rangle)\).
We comment here that the task of the device $\mathcal{M}$ is akin to leaking the quantum entity  from one of the slits (precisely, the upper one). The leaked out part  has been represented in the schematic figure as a   third channel (that is neither the upper and nor the lower path). But, we must be careful here to not create an additional dimension, wherein the middle path will be represented by a state that is orthogonal to both \(|0\rangle\) and \(|1\rangle\), which we want to avoid. This set-up may be easier, at the fundamental level, to imagine within the spin dimension of a quantum spin-1/2   system, where the fear of an extra dimension will not appear.

Getting back to description of the schematic set-up in Fig. \ref{fig1}, at this stage, 
%Therefore 
the state of the quantum system, up to a normalization, has the form $\alpha\sqrt{\mathcal{R}}|0\ket + \alpha\sqrt{1-\mathcal{R}}|+\ket + \beta|1\ket$. In this way, it has been possible to superpose the three linearly dependent states $|0\rangle$, $|+\rangle$, and $|1\rangle$. In the next step, detectors $\mathcal{D}_0$, $\mathcal{D}_+$, and $\mathcal{D}_1$ are inserted in the paths corresponding to quantum states $|0\ket$, $|+\ket$, and $|1\ket$, respectively. All the detectors are prepared in the initial state $|d\rangle$. Let the quantum system that came through the slits, and passed via the device \(\mathcal{M}\) and the Hadamard gate,  be denoted by $\mathbb{Q}$, and the quantum system of the detectors be denoted by $\mathbb{D}$. The interaction of the two systems (the detector and the one that came through the slits)
%state and quantum state 
is given by a unitary operation. It is easy to show the existence of a unitary operator $U_{\alpha,~\beta,~\mathcal{R}}$ such that the joint state of the system and detector, after the interaction, becomes $|\Psi_\mathbb{QD}\rangle=(\alpha\sqrt{\mathcal{R}}|0\rangle|d_{0}\rangle + \alpha\sqrt{1-\mathcal{R}}|+\rangle|d_{+}\rangle + \beta|1\rangle|d_{1}\rangle)/N$, where $|d_0\ket$, $|d_+\ket$, and $|d_1\ket$ are the detector states after the interaction with respective states $|0\ket$, $|+\ket$, and  $|1\ket$, and $N$ is the normalization. 

The path distinguishability is a quantity which quantifies the probability of detecting the path that the quantum system \(\mathbb{Q}\) has passed through. A possible measure of this quantity is the probability of unambiguous discrimination of the detector states corresponding to different paths. In the particular case at our hand, there exist three paths corresponding to three states $|0\ket$, $|+\ket$, and $|1\ket$ of the quantum system \(\mathbb{Q}\). Since only linearly independent states can be probabilistically distinguished unambiguously, it is 
%therefore 
necessary to choose linearly independent detector states corresponding to different paths in order to unambiguously discriminate among the paths traced by the quantum system \(\mathbb{Q}\). Therefore, let the detector states belong to the Hilbert space $\mathbb{C}^3$. 
The  state of the detector after the interaction is given by $\rho_\mathbb{D}=\text{tr}_\mathbb{Q}|\Psi_\mathbb{QD}\ket\bra \Psi_\mathbb{QD}|$. As $\rho_\mathbb{D}$ involves cross terms in the basis of detector states formed by $|d_0\ket$, $|d_+\ket$, and $|d_1\ket$, we perform a probabilistic phase damping of $\rho_\mathbb{D}$, so that the resultant state, $\rho'_\mathbb{D}$, is diagonal in basis 
%of new detector states. 
\(\{|d_0\ket, |d_+\ket, |d_1\ket\}\).
Consider the POVM given by the POVM elements $A_0^{1/2}$, $A_+^{1/2}$, $A_1^{1/2}$, $A_?^{1/2}$, where 
\begin{eqnarray}
\label{povm}
A_0&=&c(\mathbb{I}_3 - |d_+\ket\bra d_+| - |d_+^{\perp}\ket\bra d_+^{\perp}|), \nonumber \\ \nonumber A_+&=&c(\mathbb{I}_3 - |d_1\ket\bra d_1| - |d_1^{\perp}\ket\bra d_1^{\perp}|), \\ \nonumber A_1&=&c(\mathbb{I}_3 - |d_0\ket\bra d_0| - |d_0^{\perp}\ket\bra d_0^{\perp}|), \\  A_?&=&\mathbb{I}_3-A_0-A_+-A_1,
\end{eqnarray}
 with 
$c$ being the maximum positive value for which $A_?$ remains positive semidefinite. Here,
\begin{eqnarray}
|d_+^{\perp}\ket&=&(|d_1\ket-|d_+\ket\bra d_+|d_1\ket)/N_+, \nonumber \\
|d_1^{\perp}\ket&=&(|d_0\ket-|d_1\ket\bra d_1|d_0\ket)/N_1, \nonumber \\
|d_0^{\perp}\ket&=&(|d_+\ket-|d_0\ket\bra d_0|d_+\ket)/N_0, 
\end{eqnarray}
with $N_+$, $N_1$, and $N_0$ being the normalizations. 

It can be seen from the expressions of $A_0$, $A_+$, and $A_1$ that a non-zero outcome is only possible if $A_0$ acts on $|d_0\ket$, $A_+$ acts on $|d_+\ket$, and $A_1$ acts on $|d_1\ket$. All the other actions on these states by these operators return zero. Therefore, these operators can be rewritten as
$$
A_0=c|d_{+1}^\perp\ket \bra d_{+1}^\perp|, ~~
A_+=c|d_{10}^\perp\ket \bra d_{10}^\perp|, ~~
A_1=c|d_{0+}^\perp\ket \bra d_{0+}^\perp|
$$
where $|d_{ij}^\perp\ket$ is orthogonal to $|d_i\ket$ and $|d_j\ket$ for $i,j=0,+,1$.  

%The operator $A_0$ acts non-trivially only on the state $|d_0\ket$, and its action on other detector states gives zero. Therefore, if $A_0$ clicks in the POVM operation on state $\rho_D$, it tells that the quanton took path $|0\ket$. In the similar fashion if $A_+$ and $A_1$ click, then the path taken by system are $|+\ket$ and $|1\ket$, respectively.         
The POVM operator $A_?$ does not unambiguously tell about the path followed by quanton as it acts non-trivially on all three detector states. Therefore, the outcome $A_?$ is discarded. The unnormalized state after
this phase damping is given by \begin{equation}
A_0^{1/2}\rho_{\mathbb{D}} A_0^{1/2}+A_+^{1/2}\rho_{\mathbb{D}} A_+^{1/2}+A_1^{1/2}\rho_{\mathbb{D}} A_1^{1/2},
\end{equation}
using the von Neumann-L{\" u}ders postulate, which on normalization takes the form $\rho'_D=p_0 |d_{+1}^\perp\ket \bra d_{+1}^\perp| + p_+|d_{10}^\perp\ket \bra d_{10}^\perp| + p_1|d_{0+}^\perp\ket \bra d_{0+}^\perp|,$ where $p_0$, $p_+$, and $p_1$ forms a probability distribution. 
The probability of unambiguous discrimination ($P$) of the detector states, $|d_{+1}^\perp\ket$, $|d_{10}^\perp\ket$, and $|d_{0+}^\perp\ket$, appearing respectively with probabilities $p_0$, $p_+$, and $p_1$, has the following upper bound
%ed by
\cite{uqsd}:
%\cite{bera'15}
\begin{align}
P \leqslant
1-\frac{2}{3}&(\sqrt{p_0p_+}|\bra d_{+1}^\perp|d_{10}^\perp \ket| \nonumber \\ 
&+\sqrt{p_+p_1}|\bra d_{10}^\perp|d_{0+}^\perp \ket|+\sqrt{p_1p_0}|\bra d_{0+}^\perp|d_{+1}^\perp \ket|).
\end{align}

But this unambiguous discrimination of the detector states does not yield the path distinguishability of the quantum system \(\mathbb{Q}\), as one of the POVM outcomes, $A_?$, in the phase damping process is discarded. Therefore, path distinguishability is obtained when the probability of  not obtaining this outcome in the phase damping process is multiplied by $P$. The probability of the outcome $A_?$ is given by tr$(\rho_D A_?)$. Therefore, the path distinguishability $\widetilde{D}$ is given by
\begin{equation}
\widetilde{D}=(1-\text{tr}(\rho_D A_?))P.
\end{equation}

Let the state of the system $\mathbb{Q}$ be denoted by $\rho_\mathbb{Q}$, which is obtained by tracing out the detector part from the joint system-detector state $|\Psi_\mathbb{QD}\ket$. The quantum  coherence, $\widetilde{C}$, of the state $\rho_\mathbb{Q}$, with respect to the linearly dependent basis, $\{|0\ket,~|1\ket, |+\ket\}$, is then obtained by using the trace norm measure of quantum coherence, defined in Eq. (\ref{Eq8}). Note that the maximum value of the quantum coherence in this case is equal to unity.
%for the chosen basis. 
Also, the maximum value that $P$ or $\widetilde{D}$ can take is 1. We have checked numerically, via a nonlinear optimization procedure, %\textcolor{red}{(what type of numerics??)} 
that  $\widetilde{C}$, quantifying the wave nature of the quantum system \(\mathbb{Q}\), and $\widetilde{D}$, quantifying the particle nature of the same, follow the complementarity relation,
\begin{equation}
\widetilde{C}+\widetilde{D}\leqslant 1.
\end{equation} 
Interestingly, the bound is independent of $\mathcal{R}$.
Note that there does not exist such a complementarity between the quantum coherence, $\widetilde{C}$, and the probability of unambiguous discrimination of the detector states, $P$, as \(\widetilde{C}+P\) is simply \(\leqslant 2\). 
%\textcolor{red}{(does the lhs reach 2??)}

Let us now consider the set-up used in \cite{das'20}. The superposition of the quantum system was considered there in a  linearly \emph{independent} basis.  Let the coherence in this case be denoted by $\widetilde{C}_{li}$, the probability of unambiguous discrimination of the detector states by $P_{li}$, and the path distinguishability by $\widetilde{D}_{li}$. It is to be noted that in \cite{das'20}, a similar phase damping of the detector state was employed as considered here. The probability of the discarded POVM outcome was not multiplied to the probability of unambiguous discrimination of the detector states, $P_{li}$, and instead $P_{li}$ itself was considered as path distinguishability. In that article, a complementarity between the quantities $\widetilde{C}_{li}$ and $P_{li}$ was presented, and it was $\widetilde{C}_{li}+P_{li}\leqslant 1$. But in our case, there does not exist such a complementarity between $\widetilde{C}$ and $P$, as discussed earlier. Of course,  there do exist the same complementarity as here between the quantum coherence, $\widetilde{C}_{li}$, and the path distinguishability, $\widetilde{D}_{li}$, viz.
\begin{equation}
\widetilde{C}_{li}+\widetilde{D}_{li}\leqslant 1.
\end{equation}

\section{Conclusion}
In this article, 
we introduced a resource theory of quantum coherence with respect to a general ``basis'', which could possibly consist of a linearly dependent set of states that span the space under consideration, which we have referred to as a linearly dependent basis. We started by providing a definition of the free states and the free operations in a resource theory of quantum coherence with respect to such a general, possibly linearly dependent, basis. 
%linearly dependent basis.
We tried to characterize the  incoherent (free) operations by presenting a necessary condition when the basis consists of three pure qubit state vectors. 
We found that the resource theory of magic can be seen as a resource theory of quantum coherence with respect to a linearly dependent basis.
We then discussed about the quantities which can 
%satisfy certain criterion to 
be a valid quantum coherence measure in the scenario under consideration, and derived a monotonicity relation under incoherent physical maps that a certain class of such measures will satisfy. We established a connection between incoherent states with respect to a general basis and generalized measurements (POVMs). 
%We  discussed that the resource theory considered in this paper is different with the resource theory based on POVMs \cite{bruss'19}. 
We then re-examined the wave-particle duality by providing a double-slit set-up in which superposition of linearly dependent states is possible. Specifically, we showed that quantum coherence and path distinguishability maintain a complementary relation also in this general case of an arbitrary basis. The complementarity bound obtained is found to be identical with the ones uncovered earlier in the literature  in the cases where superpositions of orthogonal states and  non-orthogonal linearly independent states were considered.

\label{sec7}

\end{document}